\documentclass[letterpaper, 10 pt, conference]{ieeeconf}
\IEEEoverridecommandlockouts                         
\usepackage[font=small]{caption}

\usepackage{amsmath}
\usepackage{amssymb}
\usepackage{amsthm}
\newtheoremstyle{mystyle}%
  {\topsep}%
  {\topsep}%
  {\normalfont}%
  {}%
  {\bfseries}%
  {}%
  {.5em}%
  {}%
\theoremstyle{plain}
\newtheorem{definition}{Definition}
\newtheorem{assumption}{Assumption}
\newtheorem{lemma}{Lemma}
\newtheorem{theorem}{Theorem}

\usepackage{mathtools}
\usepackage{siunitx}

\usepackage{algorithm}
\usepackage{algpseudocode}

\usepackage{lipsum}

\usepackage{subcaption}
\usepackage{subfig}

\usepackage[short,nocomma]{optidef}

\usepackage{tikz}
\usetikzlibrary{arrows.meta,positioning}
\usepackage{pgf}
\usepackage{pgfplots}
\pgfplotsset{compat=newest}

\everymath=\expandafter{\the\everymath\displaystyle}

\title{\LARGE \bf
A Hierarchical Surrogate Model for Efficient Multi-Task\\Parameter Learning in Closed-Loop Control}

\author{Sebastian Hirt$^{1}$, Lukas Theiner$^{1}$, Maik Pfefferkorn$^{1}$, and Rolf Findeisen$^{1}$%
\thanks{$^{1}$Control and Cyber-Physical Systems Laboratory,
        Technical University of Darmstadt, Germany
        {\tt\small \{sebastian.hirt, rolf.findeisen\}@iat.tu-darmstadt.de}}%
}

\newcommand\copyrighttext{%
  \footnotesize \textcopyright \the\year{} IEEE. Personal use of this material is permitted. Permission from IEEE must be obtained for all other uses, including reprinting/republishing this material for advertising or promotional purposes, collecting new collected works for resale or redistribution to servers or lists, or reuse of any copyrighted component of this work in other works.}

\newcommand\copyrightnotice{%
\begin{tikzpicture}[remember picture,overlay]
\node[anchor=south,yshift=5pt] at (current page.south) {\fbox{\parbox{\dimexpr0.8\textwidth-\fboxsep-\fboxrule\relax}{\copyrighttext}}};
\end{tikzpicture}%
}

\begin{document}
\maketitle
\thispagestyle{empty}
\pagestyle{empty}

\copyrightnotice

\begin{abstract}
Many control problems require repeated tuning and adaptation of controllers across distinct closed-loop tasks, where data efficiency and adaptability are critical. We propose a hierarchical Bayesian optimization (BO) framework that is tailored to efficient controller parameter learning in sequential decision-making and control scenarios for distinct tasks. Instead of treating the closed-loop cost as a black-box, our method exploits structural knowledge of the underlying problem, consisting of a dynamical system, a control law, and an associated closed-loop cost function. We construct a hierarchical surrogate model using Gaussian processes that capture the closed-loop state evolution under different parameterizations, while the task-specific weighting and accumulation into the closed-loop cost are computed exactly via known closed-form expressions. This allows knowledge transfer and enhanced data efficiency between different closed-loop tasks. The proposed framework retains sublinear regret guarantees on par with standard black-box BO, while enabling multi-task or transfer learning. Simulation experiments with model predictive control demonstrate substantial benefits in both sample efficiency and adaptability when compared to purely black-box BO approaches.
\end{abstract}

\vspace{-1mm}\section{Introduction}\vspace{-1mm}\vspace{-1mm}\noindent
The work in this paper is motivated by the success of model predictive control (MPC), which is a powerful framework for controlling nonlinear systems subject to constraints \cite{rawlings2017model}. However, its practical success hinges, as in many control approaches, on well-tuned parameters, such as parameters of the MPC cost function. This presents challenges in many scenarios, where task-specific tuning to different closed-loop objectives hinders optimal long-term performance. Recent closed-loop learning approaches address these issues by combining parameterized MPC with high-level learning algorithms, such as Bayesian optimization (BO) or reinforcement learning, to improve long-term performance \cite{paulson2023tutorial, kordabad2023reinforcement, hirt2024safe, hirt2023stability}. 

We build on this idea and propose a sample-efficient, hierarchical Bayesian optimization framework for multi-task learning of controller parameters. While most existing BO-based closed-loop learning methods rely on surrogate models that learn the direct mapping from controller parameters to the closed-loop cost in a black-box fashion, they often neglect structural information about the underlying control loop. This can lead to inefficient learning, requiring many learning episodes and, consequently, many costly closed-loop experiments. In contrast, we exploit the hierarchical structure of the closed-loop learning problem, where a dynamical system evolves under a parameterized controller, and the resulting state and input trajectories are evaluated through a closed-loop cost function that aggregates stage costs over finite-length episodes. By explicitly taking this structure into account, we propose a hierarchical surrogate model that learns the parameter-dependent closed-loop behavior, while computing the total closed-loop cost per episode via known aggregation schemes. This decomposition is illustrated in Figure~\ref{fig:hierarchical_surrogate}, contrasting the proposed approach with standard black-box methods.

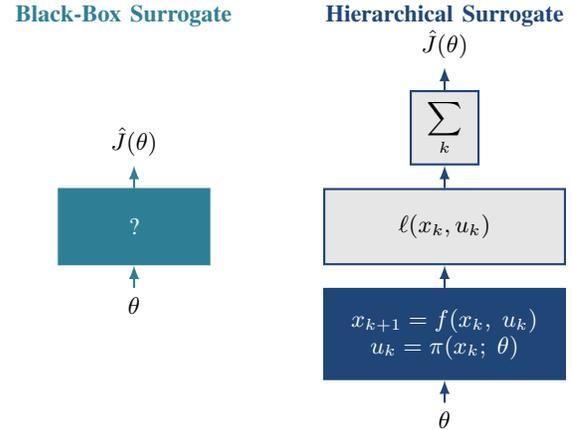
\begin{figure}[t]
    \centering
    \definecolor{turq}{HTML}{2E7F96}
\definecolor{myblue}{HTML}{1F4676}

\usetikzlibrary{arrows.meta,positioning}

\begin{tikzpicture}[%
    >=latex,
    node distance=1.4cm,
    font=\small
]

\node[draw=myblue, thick, fill=gray!20,
      minimum size=0.9cm,
      align=center] (sum) {$\sum_k$};
\node[above=0.3cm of sum] (yRight) {$\hat J(\theta)$};
\node[draw=myblue, thick, fill=gray!20,
      rectangle, minimum width=3.2cm, minimum height=1cm,
      below=0.3cm of sum, align=center] (cost)
      {$\ell(x_k, u_k)$};
\node[draw=myblue, thick, fill=myblue, text=white,
      rectangle, minimum width=3.2cm, minimum height=1.2cm,
      below=0.3cm of cost, align=center] (dyn)
      {$x_{k+1} = f(x_k,\;u_k)$\\
       $u_k = \pi(x_k;\;\theta)$};
\node[below=0.3cm of dyn] (thet) {$\theta$};
\draw[->, myblue, thick] (sum) -- (yRight);
\draw[->, myblue, thick] (cost) -- (sum);
\draw[->, myblue, thick] (thet) -- (dyn);
\draw[->, myblue, thick] (dyn) -- (cost);
\node[above=0.0cm of yRight, text=myblue, yshift=-0.2cm] (hbo) {\textbf{Hierarchical Surrogate}};

\node[draw=turq, thick, fill=turq, text=white,
      rectangle, minimum width=2.0cm, minimum height=1.0cm,
      left=1.5cm of cost] (bb) {?};
\node[below=0.3cm of bb] (G) {$\theta$};
\node[above=0.3cm of bb] (yLeft) {$\hat J(\theta)$};

\draw[->, turq, thick] (G) -- (bb);
\draw[->, turq, thick] (bb) -- (yLeft);
\node[left=1cm of hbo, text=turq] {\textbf{Black-Box Surrogate}};

\end{tikzpicture}
    \vspace{-2mm}
    \caption{Standard black-box approach (left) for controller tuning and the proposed hierarchical approach for a surrogate (right) that is tailored to parameter learning in sequential decision-making scenarios. Instead of learning the full black-box mapping between parameters and the closed-loop cost, we propose to learn the parameter-dependent closed-loop dynamics, allowing for transfer learning between different closed-loop tasks.}
    \label{fig:hierarchical_surrogate}
    \vspace{-7mm}
\end{figure}

By incorporating structure into a hierarchical Bayesian optimization procedure, we extend our previous work on time-series informed closed-loop learning \cite{hirt2024timeseriesinformed}. While our previous work exploited the time-series structure of the underlying problem via multi-fidelity modeling, it did not explicitly incorporate knowledge of the dynamical system, the parameterized controller, or the closed-loop cost function.
Consequently, we extend \cite{hirt2024timeseriesinformed} towards that in the present work. Additionally, we not only retain sublinear regret guarantees for the proposed approach, which are comparable to classical black-box Bayesian optimization, but also enable transfer learning in multi-task scenarios, where different closed-loop tasks share the same dynamical plant but vary in cost weighting.
While transfer learning and multi-task learning have been treated in the BO literature \cite{dai2020multitask, swersky2013multitask, bai2023transfer}, to the best of our knowledge, there are no methods available that tailor these approaches to closed-loop learning in sequential decision-making scenarios, exploiting this specific problem structure.
While there exist approaches that optimize control inputs directly via Bayesian methods \cite{baumgartner2025bayesian} for systems that are linear in the parameters, we focus on learning the controller parameters instead, and, additionally, allow for parameterizations that are nonlinear in the parameters. Furthermore, we exploit Gaussian process regression for learning (parts of) the BO surrogate, in contrast to \cite{baumgartner2025bayesian}. Moreover, while gray-box modeling has shown promise in Bayesian optimization \cite{paulson2022cobalt, astudillo2021thinking}, it has not been tailored to sequential decision-making with multi-task capabilities to the best of our knowledge.
The theoretical part of our work builds on literature that provides regret bounds for (multi-task) BO, e.g., \cite{srinivas2012informationtheoretic, chowdhury2021noregret, guan2024improved, dai2020multitask}, and extends it to the proposed hierarchical surrogate.

The main contributions of this work are: 1) a novel hierarchical surrogate modeling strategy that directly exploits knowledge about the hierarchical structure of the closed-loop cost function, 2) theoretical results showing sublinear cumulative regret, comparable to standard Bayesian optimization, and 3) demonstration of multi-task transfer learning that accelerates parameter learning across varied tasks. 4) We illustrate our approach with an MPC simulation, showing that significant gains in sample efficiency and adaptability to new tasks are achieved, compared to standard black-box methods from literature.

The remainder of this paper is structured as follows. In Section~\ref{sec:fundamentals}, we introduce the problem setting, including the control task and the basics of Gaussian process regression and Bayesian optimization. Section~\ref{sec:hbo} presents our hierarchical surrogate model and resulting Bayesian optimization algorithm, including the associated theoretical guarantees on sublinear regret. In Section~\ref{sec:simulation}, we validate the proposed approach through simulation experiments.

\vspace{-1mm}\section{Fundamentals}\label{sec:fundamentals}\vspace{-1mm}\noindent
In this section, we provide an overview of the control task and lay out the theoretical foundations relevant to our approach. We begin by defining the control objective and introducing parameterized model predictive control. In the following, we present Gaussian process regression as a surrogate modeling technique, followed by a description of a standard Bayesian optimization algorithm.

\vspace{-1mm}\subsection{Problem Formulation}\vspace{-1mm}\noindent
We consider nonlinear, discrete-time dynamical systems of the form 
\begin{equation}  
\label{eqn:discrete_system_general}  
    x_{k+1} = f(x_k, u_k) + w_k, 
\end{equation}  
where $ x_k \in \mathbb{R}^{n_\text{x}} $ are the system states, $ u_k \in \mathbb{R}^{n_\text{u}} $ are the control inputs, $ f: \mathbb{R}^{n_\text{x}} \times \mathbb{R}^{n_\text{u}} \to \mathbb{R}^{n_\text{x}} $ represents the (nonlinear) system dynamics, $w_k \in \mathbb{R}^{n_\text{x}}$ is Gaussian noise $w_k \sim \mathcal{N}(0, \Sigma_w)$, and $ k \in \mathbb{N}_0 $ is the discrete time index.  

Our objective is to control the system \eqref{eqn:discrete_system_general}, steering it from an initial state $ x_0 $ to a desired set-point $ (x_\mathrm{d}, u_\mathrm{d}) $ while satisfying state and input constraints. Model predictive control (MPC) provides a possible, structured framework for achieving this by solving an optimization problem at each time step based on a predefined cost function. However, the performance of MPC is highly dependent on the choice of cost function parameters, which must be tuned to achieve desirable closed-loop behavior\footnote{While we motivate our approach using MPC, the proposed framework is not limited to predictive control. It applies broadly to any control scheme with tunable parameters.}.

To address this challenge, we employ Bayesian optimization to systematically explore and refine the cost function parameters using closed-loop data. Bayesian optimization enables structured sample selection, allowing us to improve closed-loop performance in a sample-efficient manner.
Instead of treating the closed-loop performance as a black-box function, we propose a surrogate model that exploits structural knowledge about the underlying optimization problem to enhance sample efficiency and improve performance beyond standard black-box Bayesian optimization methods.

\vspace{-1mm}\subsection{Parameterized Model Predictive Control}\vspace{-1mm}\noindent
We consider a model predictive control (MPC) formulation that depends on a set of parameters $ \theta \in \Theta \subset \mathbb{R}^{n_{\text{p}}} $, where $ n_{\text{p}} \in \mathbb{N} $ denotes the number of tunable parameters. At each discrete time step $ k $, given a specific parameter choice $ \theta $, the MPC solves the parameterized optimal control problem
\begin{mini!}
    {\mathbf{\hat{u}}_k}{\left\{ \sum_{i=0}^{N-1} l_\theta({\hat x}_{i \mid k}, {\hat u}_{i \mid k}) + E({\hat x}_{N \mid k}) \! \right\}\label{eqn:mpc_ocp_cost}}{\label{eqn:mpc_ocp}}{}
    \addConstraint{\forall i}{\in \{0, 1, \dots, N-1\}: \notag}{}
    \addConstraint{}{\hat x_{i+1\mid k} = \hat f(\hat x_{i\mid k}, \hat u_{i\mid k}), \ \hat x_{0 \mid k} = x_k,}{\label{eqn:mpc_ocp_model}}
    \addConstraint{}{\hat x_{i \mid k} \in \mathcal{X}, \ {\hat u}_{i \mid k} \in \mathcal{U}, \ \hat x_{N \mid k} \in \mathcal{E}.}{\label{eqn:mpc_ocp_constraints}}
\end{mini!}
Here, $ \hat{\cdot}_{i\mid k} $ denotes the model-based prediction at time index $ k $, forecasting $ i $ steps ahead in time. The function $ \hat f: \mathbb{R}^{n_\text{x}} \times \mathbb{R}^{n_\text{u}} \to \mathbb{R}^{n_\text{x}} $ represents the prediction model, while $ x_k $ is the measured system state at time step $ k $.  
The prediction horizon has a finite length of $ N \in \mathbb{N}$. The stage and terminal cost functions are denoted by $l_\theta: \mathbb{R}^{n_\text{x}} \times \mathbb{R}^{n_\text{u}} \to [0, \infty)$ and $E: \mathbb{R}^{n_\text{x}} \to [0, \infty)$. The constraints in \eqref{eqn:mpc_ocp_constraints} include the state, input, and suitably designed terminal constraint sets \cite{rawlings2017model}, represented as $ \mathcal{X} \subset \mathbb{R}^{n_\text{x}} $, $ \mathcal{U} \subset \mathbb{R}^{n_\text{u}} $, and $ \mathcal{E} \subset \mathcal{X} $, respectively. 

Solving \eqref{eqn:mpc_ocp} yields the optimal input sequence $ \mathbf{\hat{u}}_k^*(x_k; \theta)=[\hat u_{0 \mid k}^*(x_k; \theta),\dots,\hat u_{N{-}1 \mid k}^*(x_k; \theta)] $. The first control input from this sequence is then applied to the system \eqref{eqn:discrete_system_general}. Consequently, solving \eqref{eqn:mpc_ocp} at each time step $ k $ defines the parameterized control policy $ \pi(x_k;\theta) = \hat u_{0 \mid k}^*(x_k; \theta) $.

\vspace{-1mm}\subsection{Gaussian Process Regression}\vspace{-1mm}\noindent
We explore Gaussian process (GP) regression to construct a surrogate model of the closed-loop performance as a function of the controller parameters $\theta$. Specifically, GP regression enables probabilistic modeling of an unknown function $\varphi: \mathbb{R}^{n_\xi} \to \mathbb{R}$ from observations.

Loosely speaking, a GP $g(\xi) \sim \mathcal{GP}(\mu(\xi), k(\xi, \xi'))$ is a distribution over functions and is fully specified by a prior mean function $\mu: \mathbb{R}^{n_\xi} \to \mathbb{R}, \xi \mapsto \mathrm{E}[g(\xi)]$ and a covariance function $k: \mathbb{R}^{n_\xi} \times \mathbb{R}^{n_\xi} \to \mathbb{R}, (\xi, \xi') \mapsto \mathrm{Cov}[g(\xi), g(\xi')]$, which are used to encode structural prior knowledge or beliefs, such as smoothness properties, about the unknown function.

To model the unknown function, we rely on a data set $\mathcal{D} = \{ (\xi_i, y_i) \}_{i=1}^{n_\text{d}}$ of noisy observations $y_i = \varphi(\xi_i) + \varepsilon_i$, with $\varepsilon_i \sim \mathcal{N}(0, \sigma^2)$. GP regression yields a posterior distribution over function values at test inputs $\xi_*$. 

The posterior distribution at a test point $\xi_*$ is Gaussian, $g(\xi_*) \mid \mathcal{D}, \xi_* \sim \mathcal{N}(\hat\mu(\xi_*), \hat k(\xi_*))$, with mean and variance given by
\begin{subequations}
\label{eqn:posterior_gp}
\begin{align}
    \hat\mu(\xi_*) &= \mu(\xi_*)+k(\xi_*,  \Xi) k_{y}^{-1} (y-\mu(\Xi)), \label{eqn:gp_postMean} \\
    \hat k(\xi_*) &= k(\xi_*, \xi_*) - k(\xi_*, \Xi) k_{y}^{-1} k(\Xi, \xi_*). \label{eqn:gp_postVar}
\end{align}
\end{subequations}
Here, $k_{y} = k(\Xi, \Xi) + \sigma^2 I$, $I$ is the identity matrix, and we denote the training input matrix by $\Xi = [\xi_1^\top, \dots, \xi_{n_\text{d}}^\top]^\top \in \mathbb{R}^{n_\text{d} \times n_\xi}$ and the corresponding training outputs by $y \in \mathbb{R}^{n_\text{d}}$. The posterior mean \eqref{eqn:gp_postMean} serves as a prediction for $\varphi(\xi_*)$, while the posterior variance \eqref{eqn:gp_postVar} quantifies the associated uncertainty.
The choice of mean and kernel functions, including their so-called hyperparameters, directly affects the model’s expressiveness. These hyperparameters can be inferred via evidence maximization \cite{rasmussen2006gaussian} from the training data $\mathcal{D}$.

We will later rely on the well-calibratedness of GP models in the sense of the following definition to quantify their approximation accuracy.
\begin{definition}[Well-calibrated Gaussian process \cite{krishnamoorthy2022safe}]
\label{defn:gp_calibration}
Suppose $g(\xi)$ is a posterior Gaussian process model with mean function $\hat{\mu}$ and covariance function $\hat{k}$ of a target function $\varphi: \mathbb{R}^{n_\xi} \to \mathbb{R}$.
The GP model is said to be \emph{well-calibrated} if, for all $\xi$ in the region of interest $\mathfrak{X} \subset \mathbb{R}^{n_\xi}$,
\begin{equation}
\label{eqn:gp_calibration}
    \Bigl|\,\hat \mu(\xi)\;-\;\varphi(\xi)\Bigr|
    \;\leq\;
    \beta(\delta^\prime)\,\sqrt{\hat k(\xi)} 
    \coloneq
    \epsilon(\xi; \delta^\prime)
\end{equation}
holds with probability at least $1{-}\delta^\prime$ for any $\delta^\prime\in(0,1)$ and an appropriate constant (confidence parameter) $\beta(\delta^\prime)$.
\end{definition}

A variety of results exist in the literature for determining suitable $\beta(\delta^\prime)$ such that \eqref{eqn:gp_calibration} holds. One such result, adapted from \cite{chowdhury2017kernelized}, is stated below.

\begin{lemma}
\label{lem:gp_calibration_lemma}(\cite{chowdhury2017kernelized})
Assume that the unknown function $\varphi$ lies in the reproducing kernel Hilbert space (RKHS) $\mathcal{H}_k$ associated with a continuous, positive-definite kernel $k$ and that its RKHS norm is bounded by $B \in \mathbb{R}_+$, i.e., $\|\varphi\|_k \le B$. Furthermore, assume that the GP training data are corrupted by $R$-sub-Gaussian measurement noise. Then, the GP modeling $\varphi$ is well-calibrated in the sense of Definition~\ref{defn:gp_calibration} with
\begin{equation*}
   \beta(\delta^\prime)
   =
   B + R\,\sqrt{\gamma + 1 + \ln\left({\delta^\prime}^{-1}\right)},
\end{equation*}
for any given $\delta^\prime \in (0,1)$, where $\gamma$ denotes the maximum information gain.
\end{lemma}

Several related calibration bounds can be found in the literature, e.g., in \cite{Capone2022,Fiedler2021}, and references therein. Although many of these results rely on the use of continuous, positive-definite kernels, this is typically not restrictive. 
In fact, many positive-definite kernels, such as the squared exponential and Matérn kernels, have the universal approximation property, i.e., can approximate any continuous function on a compact domain with arbitrary accuracy \cite{steinwart2008support}.
The reason is that the RKHS associated with universal kernels can be shown to be densely embedded in the space of continuous functions.

\vspace{-1mm}\subsection{Bayesian Optimization}\vspace{-1mm}\noindent
Bayesian optimization (BO) is a sample-efficient method for optimizing expensive-to-evaluate and potentially non-convex black-box functions. In the context of this work, BO is employed to optimize the closed-loop performance of system \eqref{eqn:discrete_system_general} under the parameterized feedback law given by the MPC \eqref{eqn:mpc_ocp} by learning the controller parameters $\theta \in \Theta$.
Specifically, the optimization problem is given by
\begin{equation}
    \label{eqn:bo_unconstrained_problem}
    \theta^* = \arg \min_{\theta \in \Theta} \, J(\theta),
\end{equation}
where $J: \Theta \to \mathbb{R}$ represents the unknown, expensive-to-evaluate closed-loop cost or performance function.
BO proceeds iteratively. In each iteration $t \in \mathbb{N}$, the following steps are performed:
\begin{enumerate}
    \item[1)] Select the next parameter set $\theta_t$ by optimizing a so-called acquisition function over the surrogate model.
    \item[2)] Evaluate $J(\theta_t)$ by running a closed-loop experiment or simulation.
    \item[3)] Update the dataset $\mathcal{D}_{t+1} \leftarrow \mathcal{D}_t \cup \{\theta_t, J(\theta_t)\}$ and refine the surrogate model using the new data point.
\end{enumerate}

The acquisition function $\alpha: \Theta \to \mathbb{R}$ guides the search and determines which parameter values to evaluate next. It trades off exploration (sampling in regions with high uncertainty) and exploitation (sampling near the current estimated optimum). A common acquisition function is the \textit{lower confidence bound} (LCB), given by
\begin{equation}
    \label{eqn:bo_ucb}
    \alpha(\theta; \mathcal{D}_t) = \hat\mu_{J}(\theta) - \beta_t \sqrt{\hat k_{J}(\theta)}.
\end{equation}
Here, $\hat\mu_{J}(\theta)$ and $\hat k_{J}(\theta)$ are the posterior mean and variance of the GP surrogate model of $J$, respectively, and $\beta_t > 0$ is a tunable exploration parameter. The next query point is chosen according to
\begin{equation}
    \label{eqn:bo_update_unconstrained}
    \theta_{t} = \arg \min_{\theta \in \Theta} \, \alpha(\theta; \mathcal{D}_t).
\end{equation}

The above process is repeated for a fixed number of iterations or until convergence. 
Given its sample efficiency, BO is well-suited for applications such as controller tuning, where function evaluations are expensive and limited. However, most approaches do not exploit the underlying structure of the closed-loop optimization problem. We propose a possible solution in the next section.

\vspace{-1mm}\section{Hierarchical Surrogate Model}\label{sec:hbo}\vspace{-1mm}\noindent
We explain how we exploit the hierarchical structure of the closed-loop parameter learning problem to build a tailored surrogate model and provide regret bounds for the resulting Bayesian optimization algorithm.

\vspace{-1mm}\subsection{Hierarchical Surrogate Model}\vspace{-1mm}\noindent
The proposed approach is motivated by the fact that standard black-box approaches for Bayesian-optimization-based closed-loop learning do not exploit the structure of the underlying problem. In this work, instead of learning $J(\theta)$ in a purely black-box manner, our method leverages explicit knowledge of the hierarchical closed-loop structure.
Namely, the cost function 
\begin{equation}
    \label{eqn:real_closed-loop_cost}
    J(\theta)
    \;=\;
    \sum_{k=1}^{K} 
    \ell \bigl(z_k)
\end{equation}
can be expressed as the sum of stage costs $\ell : \mathbb{R}^{n_\mathrm{x}} \times \mathbb{R}^{n_\mathrm{u}} \to [0, \infty)$ over a trajectory of length $K \in \mathbb{N}$, generated by the parameterized control policy applied to the dynamical system under control.
In \eqref{eqn:real_closed-loop_cost}, we define $z_{k+1} = (x_{k+1}, u_k) \in \mathcal{Z} := \mathcal{X} \times \mathcal{U}$ as the concatenation of the state $x_{k+1}$ and input $u_k$.

Specifically, we learn a model of the closed-loop dynamics and the control law, and then roll it out under a given parameterization to approximate the closed-loop cost \eqref{eqn:real_closed-loop_cost}. An illustration of the approach compared to using a black-box surrogate model of $J$ is shown in Fig.~\ref{fig:hierarchical_surrogate}. Throughout this section, we denote by $ \| \cdot \| $ the Euclidean (2-)norm, unless stated otherwise. We assume the following regularity for the closed-loop stage cost.
\begin{assumption}[Lipschitz Continuous Stage Cost]
\label{assump:lipschitz_stage_cost}
Let $z = (x,u)$ and $z^\prime = (x^\prime,u^\prime)$ be elements of $\Tilde{\mathcal{Z}} \supseteq \mathcal{Z}$, where $\Tilde{\mathcal{Z}}$ is the reachable set of the GP after $K$ steps, starting at the initial state $x_0$. Suppose there exists a constant $ 0 < L_\ell < \infty$ such that the stage cost $\ell$ satisfies
\begin{equation*}
\bigl|\ell(z) \;-\; \ell(z^\prime)\bigr| \;\le\; L_\ell \,\|z - z^\prime\|,
\end{equation*}
for all $z,z^\prime \in \Tilde{\mathcal{Z}}$.
\end{assumption}

\subsubsection{Surrogate via Learned Dynamics Rollout}
Let 
$$
f_{\mathrm{c}} : \mathcal{X}\times\Theta \;\to\;\mathcal{Z}, \quad (x_k, \theta) \mapsto z_{k+1}
$$ be the true but unknown closed-loop dynamics mapping, where both $\mathcal{Z}$ and $\Theta$ are compact.
In our approach, we learn an approximate closed-loop dynamics mapping 
which models $f_{\mathrm{c}}$.
Concretely, we train a separate GP for each output dimension (i.e., each component of $z_{k+1}$), which is a standard approach for GP-based system identification, e.g., \cite{hewing2020cautious}. With a slight abuse of notation, we introduce the shorthand $\hat{f}_{\mathrm{c}} \sim \mathcal{GP}(\hat\mu, \hat k)$ for the approximate multi-output mapping. Once learned, we use the GP mean $\hat\mu$ to approximate the closed-loop update
\begin{equation*}
    \hat{z}_{k+1}
    \;=\;
    \hat\mu\bigl(\hat{x}_k,\theta\bigr),
\end{equation*}
which can be iterated over the closed-loop horizon to evaluate candidate parameters $\theta$ in a computationally efficient manner, as discussed next. 

\subsubsection{Approximate Closed-Loop Cost}
Given a parameter $\theta$, we define a multi-step \emph{rollout} with the learned dynamics, starting from an initial state $\hat{x}_0 = x_0$.  
We then compute the stage-wise cost in each step and sum up along the closed-loop horizon of length $K$ according to
\begin{equation}
\label{eqn:approximate_closed-loop_cost}
    \hat{J}(\theta)
    \;=\;
    \sum_{k=1}^{K} 
    \ell \bigl( \hat z_k\bigr).
\end{equation}
Consequently, a single learned mapping $\hat{\mu}$ suffices to model
$n_\ell$ closed-loop tasks, defined by different closed-loop stage cost functions $\ell_i,~ i = 1, \ldots, n_\ell$. This enables multi-task and transfer learning without retraining the dynamics model for each new closed-loop cost function or reference. Furthermore, each closed-loop rollout automatically provides $K$ data points (state–input pairs), allowing the learned model $\hat{\mu}$ to become more accurate with each experiment across potentially different tasks.
Next, we provide a theoretical result that guarantees sublinear regret of the BO procedure based on the proposed hierarchical surrogate.

\vspace{-1mm}\subsection{Regret Bounds}\vspace{-1mm}\noindent
We operate a Bayesian optimization (BO) procedure over $T$ rounds, selecting $\theta_t$ at each round. We observe the closed-loop states and control inputs, and update $\hat{f}_{\mathrm{c}}$ accordingly. We define the cumulative regret after $T$ iterations as
\begin{equation*}
    R_T \;=\; \sum_{t=1}^T \Bigl( J(\theta_t) - J(\theta^*)\Bigr),
\end{equation*}
where $J(\theta^*) = \min_{\theta \in \Theta} \;J(\theta)$.
We now establish sublinear regret bounds for the proposed hierarchical Bayesian optimization procedure. The analysis proceeds in two steps, as common in the literature on BO regret bounds, e.g., \cite{srinivas2012informationtheoretic}. We first bound the approximation error introduced by using the learning-based surrogate $\hat{J}(\theta)$ in place of the true closed-loop cost $J(\theta)$. This is our main technical contribution. We then apply standard arguments from the literature to show that the cumulative regret grows sublinearly with the number of iterations $T$, resembling classical black-box BO in asymptotic performance while enabling transfer learning.

We rely on the following assumptions.
\begin{assumption}\label{assump:real_invariance}
    For any parameter vector $\theta \in \Theta$, the real closed-loop system satisfies $\forall (x,\theta) \in \mathcal{X}_N \times \Theta:~ f_{\mathrm{c}}(x,\theta) \in \mathcal{X}_N \times \mathcal{U}$ for a set $\mathcal{X}_N \subseteq \mathcal{X}$.
\end{assumption}
This assumption is often not restrictive, as MPC is repeatedly feasible under standard assumptions on the design of the underlying optimal control problem (OCP), see, e.g., \cite{grune2017nonlinear, rawlings2017model}.
Repeated feasibility can be established independently of the here used cost function parametrization of the underlying OCP if certain structural properties are retained and if the MPC is initially feasible.
In consequence, the set of feasible initial conditions $\mathcal{X}_N$, generally a subset of $\mathcal{X}$, is forward invariant for the closed-loop dynamics and the closed-loop inputs satisfy $u_k \in \mathcal{U}$.
Specifically, the set $\mathcal{X}_N$ is the domain of the closed-loop dynamics under the MPC.

\begin{assumption}[GP Lipschitz Continuity]
\label{assump:lipschitz}
There exists $0 < L_{\mathrm{GP}} < \infty$ such that, for all $(x,\theta),(x',\theta') \in \mathbb{R}^{n_\mathrm{x}} \times \mathbb{R}^{n_\mathrm{p}}:$ %
\begin{equation*}
    \bigl\|\hat\mu(x,\theta) - \hat\mu(x',\theta')\bigr\|
    \;\le\;
    L_{\mathrm{GP}} \,\bigl\|\,(x,\theta) - (x',\theta')\bigr\|.
\end{equation*}
\end{assumption}
Again, this assumption can often be satisfied. 
Lipschitz continuity, even globally, can be achieved by appropriate GP design, specifically by the choice of an appropriate prior mean and prior covariance function.
Technically, we further require finite observations of the underlying functions, which is naturally the case in practice, and a well-defined inverse training covariance matrix, which can be achieved by regularization.
Lipschitz continuity of the posterior mean then follows from the observation that it is represented as a linear combination of half-evaluated covariance functions, see \cite{rasmussen2006gaussian} for details.
To justify the use of a GP surrogate for the closed-loop mapping $ f_{\mathrm{c}}(x, \theta) $, we require that the GP can approximate the closed-loop mapping sufficiently well, as formalized in the following. %
\begin{assumption}
    The component functions of the closed-loop mapping $f_\mathrm{c}$ lie in the RKHS corresponding to the kernel of the GP model.
\end{assumption}
For commonly employed universal kernels, this assumption is justified if the closed-loop dynamics and the feedback law are continuous, as their RKHS is densely embedded in the space of continuous functions.
In practice, this assumption must be treated carefully in view of the continuity of the feedback law, which is not naturally given for nonlinear MPC formulations.
We refer the reader to \cite{fiacco1990} for a detailed overview of regularity properties of MPC feedback laws.

We now formulate the following lemma on the error bound for a multi-output GP model.
\begin{lemma}[Error Bound for Multi-output Model]
\label{lem:gp}
Let all assumptions from Lemma~\ref{lem:gp_calibration_lemma} hold. Let $\hat{\mu}_\tau = \begin{bmatrix} \hat{\mu}_\tau^{(1)}, \dots, \hat{\mu}_\tau^{(M)} \end{bmatrix}$ be the posterior means of the multi-output model $\hat f_{\mathrm{c}}$ after $\tau$ data points, with $M=n_\mathrm{x}{+}n_\mathrm{u}$.
If all $M$ GP models are well-calibrated according to Definition~\ref{defn:gp_calibration} for $\delta'=1-(1{-}\delta)^{1/M}$, then with probability at least $1{-}\delta$,
\begin{equation*}
    \forall (x,\theta) \in \mathcal{X}\times \Theta:
    \
    \left\|
        \hat{\mu}_\tau(x,\theta) - f_{\mathrm{c}}(x,\theta)
    \right\|
    \;\le\;
    \left\| \epsilon_\tau(x, \theta; \delta^\prime) \right\|,
\end{equation*}
holds, with $\epsilon_\tau(x, \theta; \delta^\prime) = \begin{bmatrix} \epsilon_\tau^{(1)}(x, \theta; \delta^\prime), \dots, \epsilon_\tau^{(M)}(x, \theta; \delta^\prime) \end{bmatrix}$.
\end{lemma}
\begin{proof}
    Noting that the individual GPs are independent of each other, the result is a straightforward generalization of Lemma \ref{lem:gp_calibration_lemma}. 
\end{proof}
On this basis, we now establish a bound on the multi-step error when rolling out the GP mean along a trajectory of length $K$. We form the $k$-fold composition $f_{\mathrm{c}}^k(\theta)$, with 
\begin{equation*}
    f_{\mathrm{c}}^{k+1}(\theta) \;=\; f_{\mathrm{c}}\bigl(f_{\mathrm{c}}^k(\theta)|_x,\,\theta\bigr), \quad f_{\mathrm{c}}^1(\theta) \;=\; f_{\mathrm{c}}(x_0,\theta),
\end{equation*}
where $|_x$ extracts the $x$ component from the output of $f_{\mathrm{c}}$, i.e., $(x, u)|_x = x$.
Similarly, we define the multi-step GP mean, given by
\begin{equation*}
    \hat{\mu}^{k+1}(\theta) = \hat{\mu}\bigl(\hat{\mu}^{k}(\theta)|_x,\,\theta\bigr), \quad \hat{\mu}^1(\theta) \;=\; \hat{\mu}(x_0,\theta).
\end{equation*}
\begin{lemma}[Multi-step Error Bound]
\label{lem:composition}
Suppose the assumptions from Lemma~\ref{lem:gp_calibration_lemma} as well as Assumptions~\ref{assump:real_invariance} and \ref{assump:lipschitz} hold. Then, for all $k \in \{ 1, \dots, K \}$ with probability at least $1{-}K\delta$,
\begin{equation*}
  \forall \theta\in\Theta:
  \quad
  \bigl\|\hat{\mu}^k(\theta) - f_{\mathrm{c}}^k(\theta)\bigr\|
  \le
  \nu_\tau^k(\theta)
\end{equation*}
holds, where $\nu_\tau^{k+1}(\theta)=\bigr\|\epsilon_\tau(\hat{\mu}^k|_x, \theta; \delta^\prime)\bigl\| \;+\; L_{\mathrm{GP}}\,\nu_\tau^{k}(\theta)$ and $\nu_\tau^1(\theta)=\bigl\|\epsilon_\tau(x_0, \theta; \delta^\prime)\bigr\|$.
\end{lemma}

\begin{proof}
We prove the assertion by induction.
To this end, note that for $k{=}1$, the statement is immediate since 
$\bigl\|\hat{\mu}(x,\theta) - f_{\mathrm{c}}(x,\theta)\bigr\|\le\bigl\|\epsilon_\tau(x, \theta; \delta^\prime)\bigr\|$ according to Lemma~\ref{lem:gp}.  
Now assume 
$\bigl\|\hat{\mu}^k(\theta) - f_{\mathrm{c}}^k(\theta)\bigr\|\le \nu_\tau^k(\theta)$.  
Then
\begin{align*}
&\bigl\|\hat{\mu}^{k+1}(\theta) - f_{\mathrm{c}}^{k+1}(\theta)\bigr\| \\
\;&=\,
 \ \bigl\|\hat{\mu}\bigl(\hat{\mu}^k|_x\bigr)
    -
    \hat{\mu}\bigl(f_{\mathrm{c}}^k|_x\bigr)
  +
  \hat{\mu}\bigl(f_{\mathrm{c}}^k|_x\bigr)
    -
    f_{\mathrm{c}}\bigl(f_{\mathrm{c}}^k|_x\bigr)
  \bigr\| \\
\;&\le\,\,\,
  \bigl\|
  \hat{\mu}\bigl(f_{\mathrm{c}}^k|_x\bigr)
    -
    f_{\mathrm{c}}\bigl(f_{\mathrm{c}}^k|_x\bigr)
  \bigr\| +
\ \bigl\|\hat{\mu}\bigl(\hat{\mu}^k|_x\bigr)
    -
    \hat{\mu}\bigl(f_{\mathrm{c}}^k|_x\bigr)
  \bigr\|
   \\
\;&\le\, 
\ \bigl\|\epsilon_\tau(f_{\mathrm{c}}^k|_x, \theta; \delta^\prime)\bigr\| \;+\;
L_{\mathrm{GP}}\bigl\|\hat{\mu}^k(\theta) - f_{\mathrm{c}}^k(\theta)\bigr\|,
\end{align*}
concluding the inductive step, where we drop the dependence on some arguments for an easier notation.

To prove the probability of the multi-step bound, we define the single-step event
\begin{equation*}
\scalebox{0.96}{$
  E_k =
  \left\{
     \bigl\|\hat{\mu}(f_\mathrm{c}^k|_x,\theta)-f_\mathrm{c}(f^k_\mathrm{c}|_x,\theta)\bigr\|
     \le
     \bigl\|\epsilon_\tau(f_\mathrm{c}^k|_x, \theta; \delta^\prime)\bigr\|
    \Big| E_{k-1}
  \right\}
$}
\end{equation*}
for each $k=1,\dots,K$ and $\Pr(E_0)=1$. By Lemma~\ref{lem:gp}, 
$\Pr\bigl(E_k\bigr) \ge 1{-}\delta$. 
We want all $E_k$ to occur simultaneously. Using Boole's inequality,
\begin{equation*}
  \Pr\Bigl(\bigcap_{k=1}^K E_k\Bigr)
  \ge
  1 - \sum_{k=1}^K \Pr\bigl(\overline{E_k}\bigr)
  =
  1 - K\delta,
\end{equation*}
where $\overline{E_k}$ are the complements of $E_k$.
Hence, the result holds with probability at least $1{-}K\delta$.
\end{proof}

We now employ the established multi-step error bound to bound the error between the real closed-loop cost \eqref{eqn:real_closed-loop_cost} and the approximate closed-loop cost \eqref{eqn:approximate_closed-loop_cost}.

\begin{lemma}[Multi-step Surrogate Error Bound]
\label{lem:multi_step_J}
Let Assumption~\ref{assump:lipschitz_stage_cost} and the Assumptions from Lemma~\ref{lem:composition} hold.
Then, 
\begin{equation*}
  \forall\theta\in\Theta: \quad\
  \bigl|\hat{J}(\theta)-J(\theta)\bigr|
  \;\le\;
  \sum_{k=1}^{K}
    L_{\ell}\,\nu_\tau^{k}\!(\theta),
\end{equation*}
holds with probability at least $1{-}K\delta$.
\end{lemma}

\begin{proof}
By application of the triangle inequality
$
  \bigl|\hat{J}(\theta)\;-\;J(\theta)\bigr|
  \le
  \sum_{k=1}^{K} 
     \bigl|\ell\bigl(\hat{\mu}^k(\theta)\bigr)
           - 
           \ell\bigl(f_{\mathrm{c}}^k(\theta)\bigr)\bigr|.
$
Applying Assumption~\ref{assump:lipschitz_stage_cost} yields
$
  \bigl|\ell(\hat{\mu}^k(\theta)) - \ell(f_{\mathrm{c}}^k(\theta))\bigr|
  \le
  L_{\ell}\,\bigl\|\hat{\mu}^k(\theta) - f_{\mathrm{c}}^k(\theta)\bigr\|.
$
From Lemma~\ref{lem:composition}, we have the multi-step composition error 
$\bigl\|\hat{\mu}^k(\theta) - f_{\mathrm{c}}^k(\theta)\bigr\|\le\nu_\tau^{k}(\theta)$ 
simultaneously for all $k=1,\dots,K$ with probability at least $1{-}K\delta$, concluding the proof.
\end{proof}

With this bound on the multi-step objective, we can now formulate the regret bounds following standard techniques, e.g., \cite{srinivas2012informationtheoretic}.

\begin{theorem}[Regret Bounds for Multi-step Surrogate BO]
\label{thm:regret}
Consider the lower confidence bound selection rule, and a GP dynamics surrogate with squared exponential kernel. Suppose the assumptions from Lemmas~\ref{lem:gp} and \ref{lem:multi_step_J} hold. Then, the cumulative regret satisfies with probability at least $1{-}T K \delta$
\begin{equation*}
    R_T
    \;=\;
    \sum_{t=1}^T
    \left|J(\theta_t) - J(\theta^*)\right|
    \;=\;
    \mathcal{O}\left( \sqrt{T} \cdot \log(T)^{d+1} \right).
\end{equation*}
\end{theorem}

\begin{proof}[Proof]
By Lemma~\ref{lem:multi_step_J}, we have $\forall\theta\in\Theta:$
\begin{equation}
    \label{eqn:definition_chi}
    \left|\hat{J}_t(\theta) - J(\theta)\right|
    \;\le\;
    \sum_{k=1}^{K}
    L_{\ell}\,\nu_{tK}^{k}\!\bigl(\theta\bigr)
    \;\coloneq\;
    \chi_t(\theta)
\end{equation}
with probability at least $1{-}K\delta$.
Following the approach from \cite{srinivas2012informationtheoretic}, and using the lower confidence bound selection rule given by
\begin{equation}
   \label{eqn:selection_rule_proof}
   \theta_{t} \;=\;
   \arg\min_{\theta\in\Theta}
   \left\{\hat{J}_t(\theta) - \chi_t(\theta)\right\},
\end{equation}
we bound the instantaneous regret 
$\,r_t = J(\theta_t) - J(\theta^*)$ by
$
    r_t
    \le
    J(\theta_t) - \left( \hat{J}_t(\theta_t)-\chi_t(\theta_t) \right)
    \le
    2\,\chi_t(\theta_t).
$
Here, we use that $\hat{J}_t(\theta_t)-\chi_t(\theta_t) \le \hat{J}_t(\theta^*)-\chi_t(\theta^*) \le J(\theta^*)$, where the first step follows from \eqref{eqn:selection_rule_proof} and the second step follows from \eqref{eqn:definition_chi} with probability at least $1{-}K\delta$.
Summing over $t=1,\dots,T$ gives
\begin{align*}
    R_T
    &=
    \sum_{t=1}^T r_t
    \;\le\;
    2\,\sum_{t=1}^T \chi_t(\theta_t)
    = 2 \sum_{t=1}^T \sum_{k=1}^K L_\ell \cdot \nu_{tK}^k(\theta_t) \\
    &= 2 \, L_\ell \sum_{t=1}^T \sum_{k=1}^K \sum_{k'=0}^{k-1} 
        L_{\mathrm{GP}}^{k - k' - 1} 
        \left\| \varepsilon_{tK}\left(x_{k'}, \theta_t \right) \right\|, \\
    &\leq 2 \, L_\ell \sum_{t=1}^T \sum_{k=1}^K \sum_{k'=0}^{k-1} 
        L_{\mathrm{GP}}^{k - k' - 1} 
        \left\| \beta_{tK} \right\| \cdot 
        \left\| \sqrt{ \hat{k}_{tK}(x_{k'}, \theta_t)} \right\|,
\end{align*}
where we replace the composition-based notation with the states $x_{k'}$, which are part of the GP training data.
Here, $\beta_{tK} = \begin{bmatrix}\beta_{tK}^{(1)}, \dots, \beta_{tK}^{(M)}\end{bmatrix}$ and $\hat{k}_{tK}(x_{k'}, \theta_t) = \begin{bmatrix}\hat{k}_{tK}^{(1)}(x_{k'}, \theta_t), \dots, \hat{k}_{tK}^{(M)}(x_{k'}, \theta_t)\end{bmatrix}$.
We extend the sum over $k'$ until $K{-}1$, collect the sum over $L_{\mathrm{GP}}^{k - k' - 1}$ in the constant $C_{\mathrm{L}}=1{+}L_{\mathrm{GP}}{+}L_{\mathrm{GP}}^2{+}\dots{+}L_{\mathrm{GP}}^{K{-}1}$ and apply the Cauchy-Schwarz inequality, giving
\begin{align*}
    &R_T
    \leq 2 \, L_\ell C_\mathrm{L} \sum_{t=1}^T \sum_{k'=0}^{K-1} 
        \left\| \beta_{tK} \right\| \cdot 
        \left\| \sqrt{ \hat{k}_{tK}(x_{k'}, \theta_t)} \right\| \\
    &\leq 2 \, L_\ell C_\mathrm{L} 
        \sqrt{\sum_{t=1}^T \!\sum_{k'=0}^{K-1}\! \left\| \beta_{tK} \right\|^2}
        \sqrt{\sum_{t=1}^T \!\sum_{k'=0}^{K-1}\! \left\| \sqrt{\hat{k}_{tK}(x_{k'}, \theta_t)} \right\|^2}.
\end{align*}
Expanding the norm over the $M$ GP output dimensions gives
\begin{equation*}
\scalebox{0.94}{$
    R_T
    \leq 2 \, L_\ell \, C_{\mathrm{L}}
    \sqrt{\sum_{i=1}^M \sum_{t=1}^T  K \cdot \left( \beta_{tK}^{(i)} \right)^2 } 
    \sqrt{\sum_{i=1}^M \sum_{t=1}^T \sum_{k'=0}^{K-1} \hat{k}^{(i)}_{tK}(x_{k'}, \theta_t) }.
$}
\end{equation*}
To obtain an information-theoretic bound on the regret, we aim to index the predictive variances $\hat{k}$ per step $k'$, by utilizing $\hat{k}_{tK}(x_{k'}, \theta_t) \leq \hat{k}_{(t{-}1)K{+}k'}(x_{k'}, \theta_t)$.
This step is justified by our algorithmic design: The GP model is only reconditioned after collecting data from a full $K$-step closed-loop simulation for computation efficiency. 
We upper bound the posterior variances after conditioning on data of a full trajectory (at time $tK$) with the ``single-step'' variances that would arise if the GP was reconditioned after each individual closed-loop step (at time $(t{-}1)K{+}k'$). 
Therefore, we bound
\begin{align*}
    \sum_{t=1}^T \sum_{k'=0}^{K-1} \hat{k}^{(i)}_{tK}(x_{k'}, \theta_t) 
    \leq \sum_{t=1}^T \sum_{k'=0}^{K-1} \hat{k}^{(i)}_{(t{-}1)K{+}k'}(x_{k'}, \theta_t) 
    \\
    \leq C_\gamma^{(i)}\sum_{t=1}^T \sum_{k'=0}^{K-1} \frac{1}{2}\log \left( 1 + \sigma^{-2}\hat{k}^{(i)}_{(t{-}1)K{+}k'}(x_{k'}, \theta_t)\right),
\end{align*}
where the last inequality requires $\hat{k}^{(i)}_{(t{-}1)K{+}k'}(x_{k'}, \theta_t) \leq 1$, which can be achieved by proper scaling, and $C_\gamma^{(i)} = 2/\log(1+\sigma^{-2})$ is a constant.
This allows us to employ Lemma 5.3 in \cite{srinivas2012informationtheoretic} and express the mutual information gain in terms of the sum of predictive posterior variances of the GP at the selected query points, which is upper bounded by the maximum information gain $\gamma_{TK}^{(i)}$.
The final regret bound
\begin{align*}
    R_T &\leq 2 \, L_\ell \, C_{\mathrm{L}} \cdot 
    \sqrt{ \sum_{i=1}^M \sum_{t=1}^T K \cdot \left( \beta_{tK}^{(i)} \right)^2 } \cdot 
    \sqrt{ \sum_{i=1}^M C_{\gamma}^{(i)} \cdot \gamma_{TK}^{(i)} }
\end{align*}
is written in compact form, showing the cumulative regret scaling as $\mathcal{O}\left( \sqrt{T} \log(T)^{d+1} \right)$. This follows by observing that, for each of the $M$ independent Gaussian process models used to model the system outputs, the maximum information gain satisfies $\gamma_{TK}^{(i)} = \mathcal{O}(\log(T)^{d+1})$, provided the kernel is the squared exponential kernel. This growth behavior is derived in Theorem 5 of \cite{srinivas2012informationtheoretic}. In addition, each GP confidence parameter satisfies $\beta_{tK}^{(i)}(\delta) = B + R\sqrt{\gamma_{tK}^{(i)} + 1 + \ln(1/\delta)}$ according to Lemma~\ref{lem:gp_calibration_lemma}. 

The regret probability of at least $1{-}T K \delta$ follows from applying the probability of at least $1{-}K \delta$ from Lemma~\ref{lem:multi_step_J} in each round and application of Boole's inequality over $T$ rounds.
\end{proof}

In summary, our approach matches the standard black-box approach in terms of asymptotic regret behavior. Our approach trades off (initially) higher error resulting from error propagation along the closed-loop horizon with the possibility for knowledge transfer between different closed-loop tasks. Additionally, we have more data points available compared to the standard black-box approach. In the next section, we illustrate the proposed hierarchical surrogate model for BO in closed-loop settings in a multi-task scenario.

\vspace{-1mm}\section{Simulation Results}\label{sec:simulation}\vspace{-1mm}\noindent
To demonstrate the effectiveness of the proposed hierarchical Bayesian optimization (HBO) framework, we consider a nonlinear cart-pole system as a benchmark. The system consists of a horizontally moving cart with a pole attached via a hinge and is governed by nonlinear dynamics. The system state is defined as $x = \begin{bmatrix} p & \dot{p} & \phi & \dot{\phi} \end{bmatrix}^\top,$ where $ p $ is the horizontal position of the cart, $ \dot{p} $ is the cart velocity, $ \phi $ is the pole angle relative to the upright position, and $ \dot{\phi} $ is the angular velocity of the pole.

The control strategy is based on a nominal MPC scheme that uses the true discretized system model, obtained via a fourth-order Runge–Kutta integration. The MPC controller solves an optimal control problem of the form \eqref{eqn:mpc_ocp} with a parameterized stage cost function given by
\begin{equation*}
l_\theta({\hat x}_{i \mid k}, {\hat u}_{i \mid k}) = {\hat x}_{i \mid k}^\top Q(\theta) {\hat x}_{i \mid k} + {\hat u}_{i \mid k}^\top R(\theta) {\hat u}_{i \mid k},
\end{equation*}
where the parameters $\theta \in \mathbb{R}^5$ represent the diagonal elements of the state cost matrix $Q$ and the scalar control cost $R$. These five parameters are to be learned via Bayesian optimization.

To evaluate controller performance, we define a closed-loop cost $ J(\theta) $ computed over a fixed horizon of $ K = 25 $ simulation steps. At each BO iteration, the candidate parameter vector $ \theta $ is applied to the controller and evaluated in the closed-loop. The resulting trajectory is evaluated using fixed evaluation cost matrices $ Q_{\text{cl}} $ and $ R_{\text{cl}} $, with the closed-loop cost given by
\begin{equation*}
J(\theta) = \sum_{k=0}^{K-1} x_{k+1}(\theta)^\top Q_{\text{cl}} x_{k+1}(\theta) + u_k(\theta)^\top R_{\text{cl}} u_k(\theta),
\end{equation*}
where $ x_k $ and $ u_k $ denote the state and control input at time step $ k $ resulting from the application of the MPC controller with parameters $ \theta $.

We consider a multi-task learning setting in which two distinct tasks are defined by varying the evaluation cost matrices used to compute the closed-loop cost. In Task~1, the evaluation state cost matrix is set to $ Q_{\text{cl}} = \mathrm{diag}(5, 0.1, 5, 0.1) $ and the input cost is $ R_{\text{cl}} = 0.1 $. In Task~2, the state cost is changed to $ Q_{\text{cl}} = \mathrm{diag}(4, 0.2, 4, 0.2) $, and the input cost is increased to $ R_{\text{cl}} = 0.2 $. Both tasks share the same underlying system dynamics and MPC structure, but differ in the way closed-loop controller performance is evaluated. To ensure a consistent comparison between methods and across tasks, we consider the same, single initial sample \( \theta_0 \in \Theta \) for all simulations.

The proposed hierarchical BO method is compared to a standard Bayesian optimization baseline that models $J(\theta)$ as a black-box function via a separate Gaussian process for each task, as well as to a multi-task Bayesian optimization method from the literature~\cite{swersky2013multitask}, with an expected improvement acquisition function. In contrast to these approaches, HBO constructs a shared surrogate model by learning the parameter-dependent closed-loop dynamics, which enables transfer learning across tasks.
For hierarchical BO, rollouts are performed using the GP \emph{mean only} propagation. This practical approximation is common in literature \cite{hewing2020cautious} and trades off computational simplicity against conservativeness. In contrast, the theoretical analysis relies on Lipschitz-based error bounds to guarantee regret properties. While the latter are not evaluated explicitly during rollouts in our simulation, they provide a sound basis for the method's convergence and transfer guarantees.

\begin{figure}
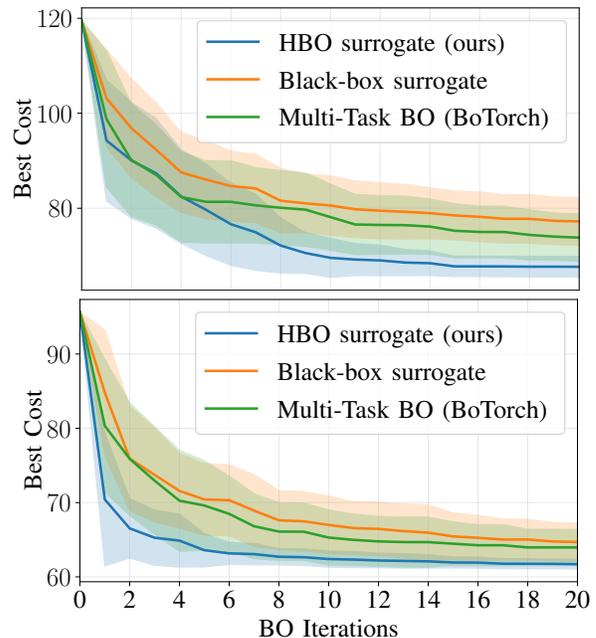

    \centering
    \vspace{0mm}
    \begin{subfigure}{0.48\textwidth}
        \centering
        \hspace{-4mm}
        \scalebox{0.48}{\input{figures/aggregate_regret_Task_1.pgf}}
    \end{subfigure}
    \begin{subfigure}{0.48\textwidth}
        \centering
        \vspace{-1.2mm}
        \scalebox{0.48}{\input{figures/aggregate_regret_Task_2.pgf}}
    \end{subfigure}
    \vspace{-8mm}
    \caption{Performance comparison between the proposed HBO approach, a standard black-box BO baseline, and a multi-task BO baseline across two tasks in the noise-free setting. Task~1 is shown at the top and Task~2 at the bottom. Shaded areas represent one standard deviation over $30$ optimization runs.}
    \vspace{-8mm}
    \label{fig:sim_multi_task}
\end{figure}
Figure~\ref{fig:sim_multi_task} illustrates the performance across both tasks in terms of the best closed-loop cost achieved over BO iterations in the noise-free case, i.e., $w_k = 0$ for all $k$. Note that we focus on closed-loop cost in the analysis of the simulation results, since the true global cost minimum is unknown and exact regret cannot be computed. In Task~1, the proposed HBO approach, the standard black-box BO, and the multi-task BO from the literature perform comparably during the initial optimization phase, indicating that incorporating structural knowledge does not degrade performance when no prior data is available. Towards the end of the optimization run, however, HBO also outperforms both baselines, which we hypothesize is due to the fact that HBO collects significantly more data per iteration by leveraging multiple state–input samples from each closed-loop rollout. Moreover, HBO exhibits a lower standard deviation over the $30$ optimization runs, indicating more consistent performance across trials. In Task~2, HBO accelerates convergence more noticeably by reusing information acquired during Task~1, as shown in Figure~\ref{fig:sim_multi_task} (bottom). This demonstrates effective transfer learning between tasks. The multi-task BO baseline also benefits from knowledge transfer, but its convergence is slower than that of HBO in our setting. In contrast, the standard black-box BO must relearn each task from scratch.
Figure~\ref{fig:sim_multi_task_noise} shows results for the noisy setting with $\Sigma_w = 0.01^2 I$ applied in the closed-loop simulation. The trends are similar to the noise-free case, and HBO continues to outperform both baselines.
Figure~\ref{fig:cumulative_regret_task_1} shows the cumulative cost incurred by each method over the optimization process. The observed sublinear growth in cumulative cost for HBO is consistent with the theoretical regret guarantees. Moreover, HBO exhibits a slower growth rate in cumulative cost compared to the two baselines, suggesting improved sample efficiency in practice. We hypothesize that this effect arises because HBO accumulates significantly more training data per iteration as each closed-loop rollout contributes multiple data points to the shared dynamics model, whereas the black-box methods update their surrogate using only a single scalar cost observation per evaluation. This richer data stream enables HBO to build a more accurate surrogate model earlier, thus additionally accelerating the learning process, even within a single task.
These results underscore the benefits of explicitly modeling the closed-loop structure when learning controller parameters across multiple related tasks.

\begin{figure}[t]
    \centering
    \vspace{2mm}
    \scalebox{0.48}{\input{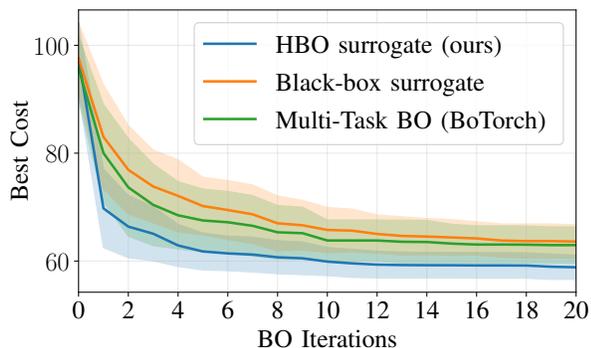}}
    \vspace{-3mm}
    \caption{Performance comparison between the proposed HBO approach, a standard black-box BO baseline, and a multi-task BO baseline for Task~2 in the noisy setting. Shaded areas represent one standard deviation over $30$ optimization runs.}
    \vspace{-6.5mm}
    \label{fig:sim_multi_task_noise}
\end{figure}

\begin{figure}[t]
    \centering
    \vspace{2mm}
    \scalebox{0.48}{\input{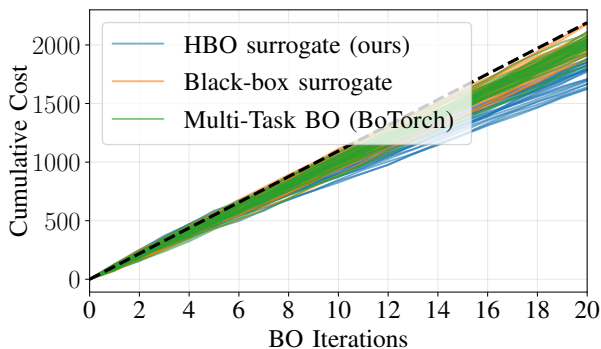}}
    \vspace{-3mm}
    \caption{Cumulative costs for Task~1 with $30$ optimization runs for each method. HBO exhibits sublinear growth, in line with our theoretical predictions. The dashed black line represents linear growth for reference.}
    \vspace{-7.5mm}
    \label{fig:cumulative_regret_task_1}
\end{figure}

\vspace{-2mm}\vspace{-1mm}\section{Conclusions}\label{sec:conclusion}\vspace{-1mm}\noindent
We proposed a hierarchical Bayesian optimization framework for efficient controller parameter learning in sequential decision-making and control. By leveraging structural knowledge of the closed-loop dynamics and closed-loop cost, our approach enables the construction of a surrogate model that generalizes across tasks, enabling transfer learning. Theoretical analysis established sublinear regret bounds comparable to standard black-box BO, while providing practical benefits through shared model updates.
Simulation results on a nonlinear cart-pole benchmark demonstrated that the proposed method matches the performance of black-box BO in single-task scenarios and significantly outperforms it in multi-task settings by sharing knowledge between different tasks. These findings highlight the advantages of embedding known system structure into the closed-loop learning process.
In future work, we aim to extend the method to enable transfer between more diverse task types and investigate higher-dimensional parameterizations. To this end, we plan to explore Bayesian neural networks for dynamics learning, enabling scalability to higher-dimensional parameter spaces.

\vspace{-1mm}\section{ACKNOWLEDGMENTS}\vspace{-1mm}\noindent
This research was supported by the German Research Foundation (DFG) within RTG 2761 LokoAssist under grant no. 450821862.

\vspace{-3mm}
\bibliographystyle{ieeetr}
\bibliography{bibliography}

\end{document}